\documentclass[conference]{IEEEtran}

\usepackage[cmex10]{amsmath}
\usepackage[pdftex]{graphicx}
\usepackage{array}
\usepackage{mdwmath}
\usepackage{mdwtab}
\usepackage{amssymb}
\usepackage{float}
\usepackage{subfig}

\newtheorem{theorem}{Theorem}

\newtheorem{example}[theorem]{Example}

\newcommand{\qed}{\hfill$\square$}

\newenvironment{proof}{%
 \noindent{\em Proof.\ }}{%
 \hspace*{\fill}\qed \\
 \vspace{2ex}}
\DeclareMathAlphabet{\bm}{OML}{cmm}{b}{it}

\newcommand{\rom}[1]{\mathrm{#1}}
\newcommand{\san}[1]{\mathsf{#1}}

\usepackage[ruled,vlined]{algorithm2e}
\usepackage{algorithmic}
\newenvironment{protocol}[1][htb]
  {%
   \begin{algorithm}[#1]%
  }{\end{algorithm}}

\begin{document}
%
\title{An Improved Lower Bound on Oblivious Transfer Capacity via Interactive Erasure Emulation}



%
\author{\IEEEauthorblockN{So Suda, Shun Watanabe, and Haruya Yamaguchi}
\IEEEauthorblockA{Department of Computer and Information Sciences, 
Tokyo University of Agriculture and Technology, Japan, \\
E-mail:s238948y@st.go.tuat.ac.jp; shunwata@cc.tuat.ac.jp; s225015t@st.go.tuat.ac.jp}}


\maketitle

\begin{abstract}
We revisit the oblivious transfer (OT) capacities of noisy channels against the passive adversary,
which have been identified only for a limited class of channels.
In the literature, the general construction of oblivious transfer has been known only for generalized
erasure channels (GECs); for other channels, we first convert a given channel to a GEC via alphabet extension 
and erasure emulation, and then apply the general construction for GEC. 
In this paper, we derive an improved lower bound on the OT capacity
of the binary symmetric channel (BSC) and binary symmetric erasure channel (BSEC) by proposing a new protocol; 
by using interactive communication between the sender and the receiver,
our protocol emulates erasure events recursively in multiple rounds.
We also discuss a potential necessity of multiple rounds interactive communication to attain the OT capacity.
\end{abstract}


%
\IEEEpeerreviewmaketitle

\section{Introduction} \label{section:introduction}

The secure computation introduced by Yao is one of the most important problems in modern cryptography \cite{Yao82}.
The goal of the parties is to compute a function in such a manner that the parties do not learn any additional information
about the inputs of other parties more than the output of the function value itself.    
For two-party secure computation, it is known that
non-trivial functions are not securely computable from scratch \cite{Bea89b, Kus92} (see also \cite{NTW15}).
In order to securely compute non-trivial functions, we usually assume that the parties can
access a primitive termed the oblivious transfer (OT) \cite{EveGolLem85}. If the OT is available, then
it is known that any functions are securely computable \cite{GolVai88, Kilian88}.

It is known that information-theoretically secure oblivious transfer can be implemented
from noisy channels between the parties \cite{CreKil88}. 
Characterizations of resources to realize secure computation is regarded as an important problem, 
and it has been actively studied in the literature
\cite{MajPraRos:09, MajPraRos12, MajPraRos13, KraQua:11, KraMajPraSah:14}.
Apart from the feasibility of secure computation, it is also interesting, from information-theoretic viewpoint, 
to characterize the efficiency of realizing secure computation from a given resource.
Such a direction of research was initiated by Nascimento and Winter in \cite{NasWin08};
they introduced the concept of the oblivious transfer (OT) capacity, which is defined as the length
of the string oblivious transfer that can be implemented per channel use.
Later, more general upper and lower bounds on the OT capacity were further studied by
Ahlswede and Csisz\'ar in \cite{AhlCsi13}. Based on the tension region idea in \cite{PraPra14},
Rao and Prabhakaran derived the state-of-the-art upper bound on the OT capacity in \cite{RaoPra14}. 

In this paper, we revisit the OT capacities of noisy channels against the passive adversary.
In the literature, the general construction of oblivious transfer has been known only for generalized
erasure channels (GECs); for other channels, we first convert a given channel to a GEC via alphabet extension 
and erasure emulation, and then apply the general construction for GEC \cite{AhlCsi13}. 
The binary symmetric channel (BSC) is a typical example that is not a GEC, and the OT capacity
is not known. Even though the binary symmetric erasure channel (BSEC) is a GEC, 
the OT capacity is not known when the erasure probability is less than half \cite{AhlCsi13}.
In this paper, we derive an improved lower bound on the OT capacity
of the BSC and BSEC by proposing a new protocol; by using interactive communication between the sender and the receiver,
our protocol emulates erasure events recursively in multiple rounds.
We numerically demonstrate that, for a certain range of parameters, 
our improved lower bound almost matches with the upper bound (but still has a gap). 

For an example of BSEC-like channel, we show that a protocol based on the same idea as above
attains the OT capacity; we also argue that the OT capacity of this example is unlikely to be attained by
one-round protocols, which suggest a potential necessity of multiple rounds interactive communication to attain the OT capacity.

The rest of the paper is organized as follows.
In Section \ref{section:problem}, we introduced the problem formulation of oblivious transfer and the OT capacity.
In Section \ref{section:BSEC}, we review the standard protocol to realize OT from BSEC;
in the process of reviewing the standard protocol, we explain the building blocks that are also used in our new protocol.
Then, in Section \ref{section:recursive}, we present our new protocol to realize OT from BSEC (and BSC as a special case).
Finally, in Section \ref{section:discussion}, we discuss a potential necessity of multiple rounds interactive communication to attain the OT capacity.


\section{Problem Formulation} \label{section:problem}

We mostly follow the notations from \cite{csiszar-korner:11, TyaWat:book}.
Random variables and random vectors are denoted by capital letters such as $X$ or $X^n$; 
realizations are denoted by lowercase letters such as $x$ or $x^n$; 
ranges are denoted by corresponding calligraphic letters such as ${\cal X}$ or ${\cal X}^n$;
for an integer $n$, the index set is denoted as $[n] = \{1,\ldots,n\}$; for a subset ${\cal I} \subset [n]$
and a random vector $X^n = (X_1,\ldots,X_n)$ of length $n$, the random vector $X_{{\cal I}}$ is a collection
of $X_i$s such that $i \in {\cal I}$. The entropy and conditional entropies are denotes as $H(X)$ and $H(X|Y)$;
the binary entropy function is denoted as $H(q)= q\log \frac{1}{q}+(1-q)\log\frac{1}{(1-q)}$ for $0 \le q \le 1$.

Let $W$ be a channel from a finite input alphabet ${\cal X}$ to a finite output alphabet ${\cal Y}$.
By using the channel $n$ times, the parties, the sender ${\cal P}_1$ and the receiver ${\cal P}_2$,
shall implement the string oblivious transfer (OT) of length $l$. More specifically, the sender generates two
uniform random strings $K_0$ and $K_1$ on $\{0,1\}^l$, and the receiver generates a uniform bit $B$ on $\{0,1\}$,
as inputs to OT protocol. In the protocol, in addition to communication over the noisy channel $W^n$,
the parties are allowed to communicate over the noiseless channel, possibly multiple rounds.\footnote{In general,
the communication over the noiseless channel may occur between invocations of the noisy channels; for a more precise 
description of OT protocols over noisy channels, see \cite{AhlCsi13}.}
Let $\Pi$ be the exchanged messages over the noiseless channel, and let $X^n$ and $Y^n$ be the input
and the output of the noisy channel $W^n$. At the end of the protocol, ${\cal P}_2$ computes an estimate
$\hat{K} = \hat{K}(Y^n,\Pi, B)$ of $K_B$. 

For $0 \le \varepsilon,\delta_1,\delta_2 <1$,
we define that a protocol realizes $(\varepsilon,\delta_1,\delta_2)$-secure OT of length $l$ (for passive adversary) if
the following three requirements are satisfied:
\begin{align}
\Pr(\hat{K} \neq K_B) &\le \varepsilon, \label{eq:correct} \\
d_{\mathtt{var}}(P_{K_{\overline{B}} Y^n \Pi B}, P_{K_{\overline{B}}} \times P_{Y^n \Pi B}) &\le \delta_1, \label{eq:security-1} \\
d_{\mathtt{var}}(P_{BK_0K_1X^n\Pi}, P_B \times P_{K_0 K_1 X^n \Pi}) &\le \delta_2, \label{eq:security-2}
\end{align}
where $d_{\mathtt{var}}(P,Q)=\frac{1}{2}\sum_a|P(a)-Q(a)|$ is the variational distance. 
The requirement \eqref{eq:correct} is referred to as $\varepsilon$-correctness; the requirement \eqref{eq:security-1}
is the security for ${\cal P}_1$ in the sense that $K_{\overline{B}}$ is concealed from ${\cal P}_2$ observing $(Y^n,\Pi, B)$,
where $\overline{B}=B\oplus 1$;
and the requirement \eqref{eq:security-2} is the security for ${\cal P}_2$ in the sense that $B$ is concealed from 
${\cal P}_1$ observing $(K_0,K_1,X^n,\Pi)$.

A rate $R$ is defined to be achievable if, for every $0 \le \varepsilon,\delta_1,\delta_2 <1$ and sufficiently large $n$,
there exists an $(\varepsilon,\delta_1,\delta_2)$-secure OT protocol of length $l$ satisfying $\frac{l}{n}\ge R$.
Then, the OT capacity $C_{\mathtt{OT}}(W)$ is defined as the supremum of achievable rates. 

The OT capacity is known only for a limited class of channels: the erasure channel, and the generalized erasure channel (GEC)
for erasure probability larger than $1/2$; we will discuss a special instance of GEC, the binary symmetric erasure channel,
in more detail in Section \ref{section:BSEC}.
The following upper bounds on the OT capacity are known in the literature:
\begin{align}
C_{\mathtt{OT}}(W) &\le \max_{P_X} \min_{P_{J|X}}\big[ I(X \wedge J|Y) + I(X \wedge Y|J) \big] \label{eq:upper-1} \\
&\le \max_{P_X}\min\big[ I(X\wedge Y), H(X|Y) \big], \label{eq:upper-2}
\end{align}
where the minimization in \eqref{eq:upper-1} is taken over channel $P_{J|X}$
from ${\cal X}$ to an auxiliary alphabet ${\cal J}$ satisfying $|{\cal J}| \le |{\cal X}||{\cal Y}| +2$.
The upper bound \eqref{eq:upper-2} was derived in \cite{AhlCsi13}, 
and the improved upper bound \eqref{eq:upper-1} was derived in \cite{RaoPra14}.

\section{Review of Standard Protocol for BSEC} \label{section:BSEC}

\begin{protocol}[h] 
\caption{Standard Protocol for BSEC} \label{protocol-BSEC}
\begin{algorithmic}[1]
\STATE \label{protocol-BSEC-step1}
${\cal P}_1$ sends the uniform bit string $X^n$ over the BSEC $W^n$, and ${\cal P}_2$ receives $Y^n$.

\STATE \label{protocol-BSEC-step2}
For each $i\in[n]$, ${\cal P}_2$ generates $V_i \in \{0,1,2\}$ as follows:
\begin{itemize}
\item If $Y_i = \san{e}$, then set $V_i=1$;

\item 
If $Y_i \in \{0,1\}$, then set $V_i=0$ with probability $\frac{p}{1-p}$ and set $V_i=2$ with probability $\frac{1-2p}{1-p}$.
\end{itemize}
Then, ${\cal P}_2$ sets $\tilde{{\cal I}}_b = \{ i \in [n]: V_i = b\}$ for $b=0,1$. If 
$|\tilde{{\cal I}}_0| < m$ or $|\tilde{{\cal I}}_1| < m$, then abort the protocol. Otherwise, ${\cal P}_2$ sets
${\cal I}_B$ as the first $m$ indices from $\tilde{{\cal I}}_0$ and ${\cal I}_{\overline{B}}$ as the first $m$ indices from $\tilde{{\cal I}}_1$,
and sends $\Pi_1 = ({\cal I}_0, {\cal I}_1)$ to ${\cal P}_1$.

\STATE \label{protocol-BSEC-step3}
${\cal P}_1$ randomly picks functions $F:\{0,1\}^m \to \{0,1\}^l$ and $G:\{0,1\}^m \to \{0,1\}^\kappa$
from universal hash families, computes $S_b = F(X_{{\cal I}_b})$ and $C_b = G(X_{{\cal I}_b})$ for $b=0,1$,
and sends $\Pi_2 = (\Pi_{2,0}, \Pi_{2,1}, \Pi_{2,2})$, where 
\begin{align*}
\Pi_{2,b} = K_b \oplus S_b \mbox{ for } b=0,1, \mbox{ and } \Pi_{2,2}=(C_0,C_1).
\end{align*}

\STATE \label{protocol-BSEC-step4}
${\cal P}_2$ reproduces $\hat{X}_{{\cal I}_B}$ from $C_B$ and $Y_{{\cal I}_B}$, and computes $\hat{K}= \Pi_{2,B} \oplus F(\hat{X}_{{\cal I}_B})$.
\end{algorithmic}
\end{protocol}

In this section, we review the standard OT protocol for binary symmetric erasure channel (BSEC) 
with erasure probability $0< p \le \frac{1}{2}$ and crossover probability $0 < q <1$ given by
\begin{align*}
& W(0|0) = W(1|1) = (1-p)(1-q), \\
& W(1|0) = W(0|1) = (1-p)q, ~
W(\san{e}|0) = W(\san{e}|1) = p.
\end{align*}
Before describing the OT protocol, let us review two building blocks known as 
the information reconciliation and the privacy amplification (eg.~see \cite[Chapters 6 and 7]{TyaWat:book} for more detail):
\paragraph*{Information Reconciliation} We consider the situation such that the sender and the receiver observe correlated i.i.d. sources $X^m$ and $Y^m$;
the sender transmits a message $\Pi$ to the receiver, and the receiver reproduces an estimate of $X^m$ by using the transmitted message and the side-information
$Y^m$. This problem is known as the source coding with side information (Slepian-Wolf coding with full side information), and it is known that
the receiver can reproduce $X^m$ with small error probability if, for some margin $\Delta>0$,  
the sender transmit a message of length $\kappa= \lceil m(H(X|Y)+\Delta)\rceil$ that is created by the universal hash family.
For instance, when $Y$ is the output of BSC with uniform input $X$ and the crossover probability is $q$, then it suffices for the sender to
transmit a message of length $\lceil m  (H(q)+\Delta) \rceil$.

\paragraph*{Privacy Amplification} It is a procedure to distill a secret key from a randomness that is partially known to the adversary.
In the OT construction, we consider the situation such that ${\cal P}_1$ and ${\cal P}_2$ (who plays the role of the adversary)
observes correlated i.i.d. sources $X^m$ and $Z^m$; furthermore, the adversary may observe additional message $\Pi$ (obtained 
during the information reconciliation). Then, for a randomly chosen function $F$ from universal hash family, the key $K=F(X^m)$ 
is almost uniform and independent of the adversary's observation $(Z^m,\Pi)$ provided that the length of the generated key 
is $l = \lfloor m(H(X|Z)- \Delta)\rfloor - \kappa$ for some margin $\Delta>0$, where $\kappa$ is the length of additional message $\Pi$.
Particularly, in the context of OT construction from BSEC, we use the privacy amplification for the case where $X$ is uniform
bit and $Z$ is constant; in that case, it suffices to set $l = \lfloor m(1-\Delta) \rfloor - \kappa$ to distill a secure secrete key.

The high-level flow of the standard protocol is as follows. 
First, ${\cal P}_1$ sends the uniform bit string $X^n$ over the BSEC 
and ${\cal P}_2$ receives $Y^n$. Then, ${\cal P}_2$ picks subsets ${\cal I}_0, {\cal I}_1 \subset [n]$ so that ${\cal I}_B$ consists of indices
without erasure and ${\cal I}_{\overline{B}}$ consists of indices with erasure. Since the erasure occurs obliviously to ${\cal P}_1$,
revealing the index sets $({\cal I}_0,{\cal I}_1)$ to ${\cal P}_1$ does not leak any information about $B$. 
On the other hand, by using randomly chosen functions from universal hash families, ${\cal P}_1$ generates a pair of secret keys
$(S_0,S_1)$, and sends $(K_0, K_1)$ to ${\cal P}_2$ by encrypting with the generated keys so that ${\cal P}_2$ can only decrypt $K_B$. 
Formally, the standard protocol is described in Protocol \ref{protocol-BSEC}. 
Since $0< p \le \frac{1}{2}$, part of non-erasure indices are discarded.\footnote{For $\frac{1}{2} < p <1$, 
we need to modify the discarding rule in Step \ref{protocol-BSEC-step2}.}
In the protocol, the random variable $V_i$ describes that the index $i$
is erasure if $V_i=1$; it is non-erasure and not discarded if $V_i=0$; and it is non-erasure but is discarded if $V_i=2$.

Now, we outline the security and the performance of Protocol \ref{protocol-BSEC}.
First, to verify that ${\cal P}_2$'s message $\Pi_1=({\cal I}_0,{\cal I}_1)$ does not leak any information about $B$ to ${\cal P}_1$,
let us introduce $\tilde{V}_i = V_i \oplus B$ if $V_i \in \{0,1\}$ and $\tilde{V}_i = V_i$ if $V_i=2$.
Then, note that ${\cal I}_0$ and ${\cal I}_1$ are the first $m$ indices of $\{ i : \tilde{V}_i=0\}$
and $\{ i : \tilde{V}_i=1\}$, i.e., $({\cal I}_0,{\cal I}_1)$ is a function of $\tilde{V}^n$.
Thus, it suffices to show that $I(X^n, \tilde{V}^n \wedge B)=0$. Since $I(X^n \wedge B) = 0$,
we will verify $I(\tilde{V}^n \wedge B|X^n)=0$.
Since $0 < p \le \frac{1}{2}$, note that
\begin{align*}
\Pr( V_i = 0 | X_i = x) = \Pr(V_i = 1 | X_i = x) = p
\end{align*}
for every $i \in [n]$, which implies
\begin{align*}
P_{\tilde{V}_i | X_i B}(v|x,0) &= P_{V_i|X_i}(v|x) \\
&= P_{V_i|X_i}(v \oplus 1|x) 
= P_{\tilde{V}_i|X_i B}(v|x,1)
\end{align*}
for $v \in \{0,1\}$. Also, we have $P_{\tilde{V}_i|X_i B}(2|x,b) = P_{V_i|X_i}(2|x)$. Thus, 
we have $I(\tilde{V}^n \wedge B|X^n) = 0$.

For a small margin $\Delta >0$, if we set $m = \lfloor n(p-\Delta) \rfloor$, then
the protocol is not aborted in Step \ref{protocol-BSEC-step2} with high probability.
Furthermore, since the crossover probability between $X_i$ and $Y_i$ conditioned on $V_i=0$ is
\begin{align*}
P_{Y_i|X_i V_i}(1|0,0) = P_{Y_i|X_i V_i}(0|1,0) = q,
\end{align*}
if we set $\kappa = \lceil m(H(q) + \Delta) \rceil$, then, by the result on the information reconciliation,
${\cal P}_2$ can reproduce $X_{{\cal I}_B}$ with small error probability in Step \ref{protocol-BSEC-step4}.
Finally, since $Y_i=\san{e}$ for $i\in {\cal I}_{\overline{B}}$,
if we set $l = \lfloor m(1-\Delta) \rfloor - \kappa$, then, by the result on the privacy amplification,
the key $S_{\overline{B}}$ is almost uniform and independent of ${\cal P}_2$'s observation;
since $K_{\overline{B}}$ is encrypted by the one-time pad with key $S_{\overline{B}}$, $K_{\overline{B}}$ is not leaked to ${\cal P}_2$.
Consequently, 
by taking $\Delta>0$ sufficiently small and $n$ sufficiently large, Protocol \ref{protocol-BSEC} realizes
$(\varepsilon,\delta_1,\delta_2)$-secure OT of length roughly $np(1-H(q))$, i.e., we have\footnote{For 
$\frac{1}{2}\le p \le 1$, it is known that $C_{\mathtt{OT}}(W) = (1-p) (1- H(q))$.}
\begin{align} \label{eq:OT-capacity-BSEC}
C_{\mathtt{OT}}(W) \ge p (1- H(q)).
\end{align}
In the next section, we improve on this lower bound
by introducing a recursive protocol.

\section{Recursive Protocol} \label{section:recursive}

\begin{figure}[tb]
\centering{
\includegraphics[width=0.8\linewidth]{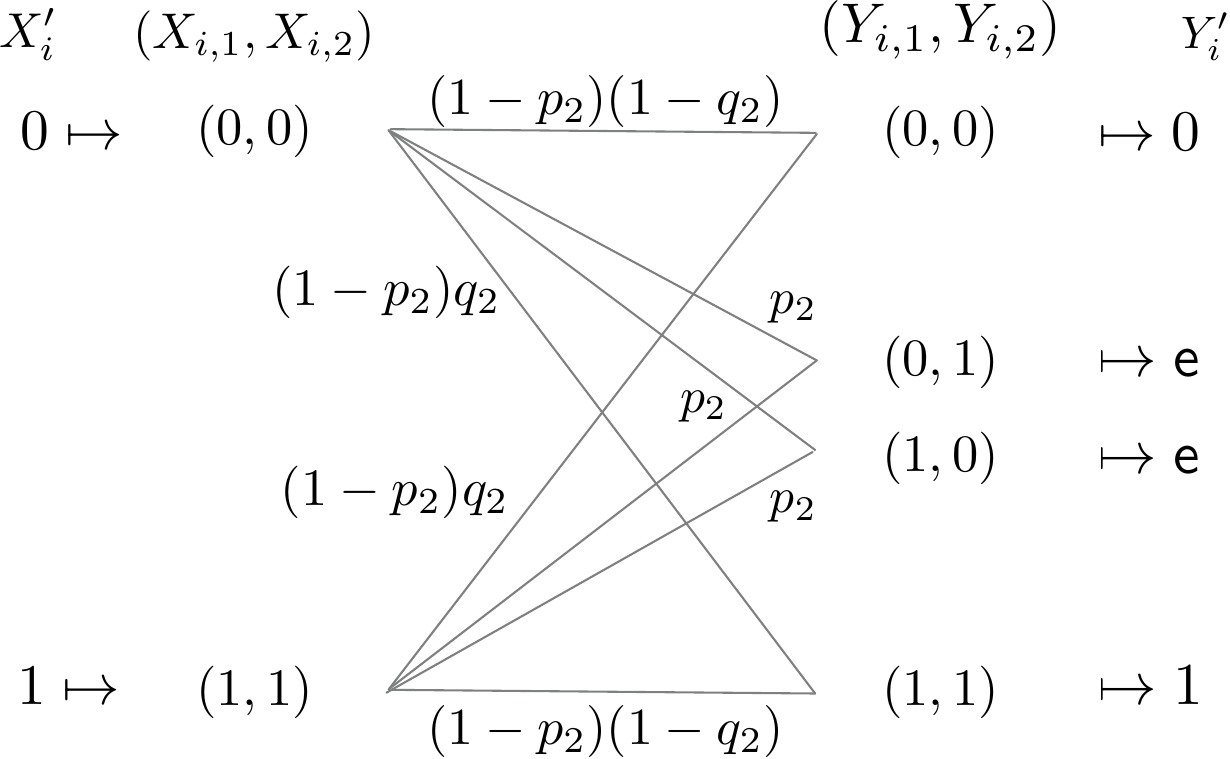}
\caption{A description emulated BSEC induced from the BSC; the erasure probability $p_2$ and the crossover probability $q_2$
are computed by \eqref{eq:erasure-crossover-probability-1}.}
\label{Fig:emulated-BSEC}
}
\end{figure}

Let us consider the binary symmetric channel (BSC) with crossover probability $0 < q_1 < 1$. 
Since there is no erasure symbol in the BSC, Protocol \ref{protocol-BSEC} cannot be used directly.
In order to apply Protocol \ref{protocol-BSEC} to the BSC, the procedures of alphabet extension and erasure emulation
have been used in the literature \cite[Example 1]{AhlCsi13}. More specifically, the parties divide the index set $[n]$ into blocks of length $2$ 
(for simplicity, assume that $n$ is even number). Then, for $i$th block, ${\cal P}_1$ randomly transmit 
$(X_{i,1},X_{i,2})=(0,0)$ or $(X_{i,1},X_{i,2})=(1,1)$. If ${\cal P}_2$ receives $(Y_{i,1},Y_{i,2}) \in \{(0,1), (1,0)\}$,
then ${\cal P}_2$ regards the received symbols as erasure $\mathsf{e}$
since posterior probability distribution of $(0,0)$ and $(1,1)$ is uniform; otherwise, ${\cal P}_2$ treats
the received symbols $(Y_{i,1},Y_{i,2})$ as they are. By relabeling $(0,0) \mapsto 0$, $(1,1) \mapsto 1$,
$(0,1) \mapsto \mathsf{e}$, and $(1,0)\mapsto \mathsf{e}$, we obtain an emulated BSEC
from $X_i^\prime$ to $Y_i^\prime$ (cf.~Fig.~\ref{Fig:emulated-BSEC}):
\begin{align*}
P_{Y_i^\prime|X_i^\prime}(0|0)=P_{Y_i^\prime|X_i^\prime}(1|1) &= (1-p_2)(1-q_2), \\
P_{Y_i^\prime|X_i^\prime}(1|0)=P_{Y_i^\prime|X_i^\prime}(0|1) &= (1-p_2)q_2, \\
P_{Y_i^\prime|X_i^\prime}(\san{e}|0)=P_{Y_i^\prime|X_i^\prime}(\san{e}|1) &= p_2,
\end{align*}
where 
\begin{align}
p_2 = 2q_1 (1-q_1) ,
~~~q_2 = \frac{q_1^2}{(1-q_1)^2 + q_1^2}. \label{eq:erasure-crossover-probability-1}
\end{align}
Then, we can apply Protocol \ref{protocol-BSEC} to this emulated BSEC; since the effective block-length is $\frac{n}{2}$,
we can derive the following lower bound on the OT capacity of the BSC $W_{\rom{BSC}(q_1)}$:
\begin{align} \label{eq:OT-capacity-BSC-1round}
C_{\mathtt{OT}}(W_{\rom{BSC}(q_1)}) \ge \frac{p_2}{2}(1- H(q_2)).
\end{align}

The idea of our recursive protocol is to conduct the above mentioned erasure emulation in multiple rounds
using interactive communication. To fix an idea, let us again consider the BSEC with erasure probability $0 < p_1 \le \frac{1}{2}$
and crossover probability $0 < q_1 <1$. We can directly apply Protocol \ref{protocol-BSEC};
however, there are remaining indices $\tilde{{\cal I}}_2 = \{ i : V_i =2\}$ that are not used in Protocol \ref{protocol-BSEC}. 
Let ${\cal I}_2 \subset \tilde{{\cal I}}_2$ be such that $|{\cal I}_2|$ is an even number exceeding $n(1-2p_1 -\Delta)$ for some margin $\Delta>0$;
we can take such ${\cal I}_2$ with high probability for sufficiently large $n$.
Note that, conditioned on $i \in {\cal I}_2$, the channel from $X_i$ to $Y_i$ is the BSC with crossover probability $q_1$.
In our recursive protocol, the parties divide the index set ${\cal I}_2$ into blocks of length $2$.
Then, for $i$th block, ${\cal P}_1$ reveals the parity $X_{i,1} \oplus X_{i,2}$ to ${\cal P}_2$; if the parity is $1$, ${\cal P}_1$
and ${\cal P}_2$ flip the value of $X_{i,1}$ and $Y_{i,1}$, respectively.
By the same relabeling as above, we obtain an emulated BSEC from $X_i^\prime$ to $Y_i^\prime$ with erasure probability $p_2$ and crossover 
probability $q_2$ given by \eqref{eq:erasure-crossover-probability-1}. Then, we can apply Protocol \ref{protocol-BSEC} to those additional indices to realize an additional string OT. 
Since the effective block-length of the additional part is $\frac{|{\cal I}_2|}{2} \ge \frac{n(1-2p_1-\Delta)}{2}$, 
we can derive the following improved lower bound on the OT capacity of BSEC $W_{\rom{BSEC}(p_1,q_1)}$:
\begin{align}
& C_{\mathtt{OT}}(W_{\rom{BSEC}(p_1,q_1)}) \label{eq:OT-capacity-BSEC-2round} \nonumber \\
&\ge p_1 (1- H(q_1)) + \frac{(1-2p_1)}{2} \cdot p_2 (1-H(q_2)).
\end{align}

\begin{protocol}[h] 
\caption{Recursive Protocol for BSEC} \label{protocol-recursive}
\begin{algorithmic}[1]
\STATE \label{protocol-recursive-step1}
${\cal P}_1$ sends the uniform bit string $(X^{(1)}_1,\ldots,X^{(1)}_n)$ 
over the BSEC $W^n$, and ${\cal P}_2$ receives $(Y^{(1)}_1,\ldots,Y^{(1)}_n)$.

\STATE \label{protocol-recursive-step2}
Let ${\cal I}^{(0)} = [n]$, and the parties conduct Steps \ref{protocol-recursive-step3}-\ref{protocol-recursive-step6}
for $t=1,\ldots,T$ (Step \ref{protocol-recursive-step6} is not conducted for $t=T$).

\STATE \label{protocol-recursive-step3}
For each $i\in {\cal I}^{(t-1)}$, ${\cal P}_2$ generates $V^{(t)}_i \in \{0,1,2\}$ as follows:
\begin{itemize}
\item If $Y^{(t)}_i = \san{e}$, then set $V^{(t)}_i=1$;

\item 
If $Y^{(t)}_i \in \{0,1\}$, then set $V^{(t)}_i=0$ with probability $\frac{p_t}{1-p_t}$ and set $V^{(t)}_i=2$ with probability $\frac{1-2p_t}{1-p_t}$.
\end{itemize}
Then, ${\cal P}_2$ sets $\tilde{{\cal I}}^{(t)}_b = \{ i \in {\cal I}^{(t-1)}: V^{(t)}_i = b\}$ for $b=0,1,2$. If 
$|\tilde{{\cal I}}^{(t)}_0| < m_t$ or $|\tilde{{\cal I}}^{(t)}_1| < m_t$
or $|\tilde{{\cal I}}^{(t)}_2|< 2 n_t$, then abort the protocol. Otherwise, ${\cal P}_2$ sets
${\cal I}^{(t)}_B$ as the first $m_t$ indices from $\tilde{{\cal I}}^{(t)}_0$, ${\cal I}^{(t)}_{\overline{B}}$ as the first $m_t$ indices from $\tilde{{\cal I}}^{(t)}_1$,
and ${\cal I}^{(t)}_2$ as the first $2n_t$ indices from $\tilde{{\cal I}}^{(t)}_2$,
and sends $\Pi^{(t)}_1 = ({\cal I}^{(t)}_0, {\cal I}^{(t)}_1, {\cal I}^{(t)}_2)$ to ${\cal P}_1$.

\STATE \label{protocol-recursive-step4}
${\cal P}_1$ randomly picks functions $F^{(t)}:\{0,1\}^{m_t} \to \{0,1\}^{l_t}$ and $G^{(t)}:\{0,1\}^{m_t} \to \{0,1\}^{\kappa_t}$
from universal hash families, computes $S^{(t)}_b = F^{(t)}(X^{(t)}_{{\cal I}^{(t)}_b})$ and $C_b^{(t)} = G^{(t)}(X^{(t)}_{{\cal I}^{(t)}_b})$ for $b=0,1$,
and sends $\Pi^{(t)}_2 = (\Pi^{(t)}_{2,0}, \Pi^{(t)}_{2,1}, \Pi^{(t)}_{2,2})$, where, for $b=0,1$, 
\begin{align*}
\Pi^{(t)}_{2,b} &= K^{(t)}_b \oplus S^{(t)}_b 
\mbox{ and } \Pi^{(t)}_{2,2}=(C^{(t)}_0,C^{(t)}_1).
\end{align*}

\STATE \label{protocol-recursive-step5}
${\cal P}_2$ reproduces $\hat{X}^{(t)}_{{\cal I}^{(t)}_B}$ from $C^{(t)}_B$ and $Y^{(t)}_{{\cal I}_B}$, and computes 
$\hat{K}^{(t)}= \Pi^{(t)}_{2,B} \oplus F^{(t)}(\hat{X}^{(t)}_{{\cal I}^{(t)}_B})$.

\STATE \label{protocol-recursive-step6}
The parties divide the indices ${\cal I}^{(t)}_2$ into blocks of length $2$. For $i$th block, 
${\cal P}_1$ reveals the parity $X^{(t)}_{i,1} \oplus X^{(t)}_{i,2}$ to ${\cal P}_2$; if the parity is $1$,
${\cal P}_1$ and ${\cal P}_2$ flip the value of $X^{(t)}_{i,1}$ and $Y^{(t)}_{i,1}$, respectively.
By relabeling $(0,0) \mapsto 0$, $(1,1) \mapsto 1$,
$(0,1) \mapsto \mathsf{e}$, and $(1,0)\mapsto \mathsf{e}$, the parties create 
sequences $(X^{(t+1)}_i : i \in {\cal I}^{(t)})$ and $(Y^{(t+1)}_i : i \in {\cal I}^{(t)})$ with 
index set ${\cal I}^{(t)}$ of length $|{\cal I}^{(t)}| = n_t$. Increment $t \to t+1$, and go to Step \ref{protocol-recursive-step3}. 

\end{algorithmic}
\end{protocol}

More generally, in $T$ rounds protocol, we recursively emulate the BSEC with 
erasure probability $p_{t+1}$ and crossover probability $q_{t+1}$ given by
\begin{align*}
p_{t+1} = 2 q_t (1-q_t),~~~q_{t+1} = \frac{q_t^2}{(1-q_t)^2 + q_t^2}
\end{align*}
for $t=1,\ldots,T-1$. Set $n_0=n$, and set
$n_t = \lceil n_{t-1}(1-2p_t-\Delta)/2 \rceil$ for $t=1,\ldots,T$.
Furthermore, set $m_t = \lceil n_{t-1}(p_t-\Delta)\rceil$, $\kappa_t=\lceil m_t(H(q_t)+\Delta)\rceil$, and $l_t = \lfloor m_t(1-\Delta) \rfloor -\kappa_t$
for $t=1,\ldots,T$.
In our protocol, the parties seek to realize string OT of length $l_t$ in each round, where
${\cal P}_2$'s input $B$ is the same in each round. 
The detail of our recursive protocol is described in Protocol \ref{protocol-recursive} ($n$ is assumed to be even number),
and its performance is given as follows.
\begin{theorem} \label{theorem:recursive}
By using Protocol \ref{protocol-recursive} with $T$ rounds, we can derive the following lower bound
on the OT capacity of BSEC $W_{\rom{BSEC}(p_1,q_1)}$:\footnote{For $t=1$, $\prod_{j=1}^{t-1} \frac{(1-2p_j)}{2}$ is regarded as $1$.}
\begin{align*}
 C_{\mathtt{OT}}(W_{\rom{BSEC}(p_1,q_1)}) 
\ge \sum_{t=1}^T \bigg\{ \prod_{j=1}^{t-1} \frac{(1-2p_j)}{2} \bigg\} p_t(1-H(q_t)).
\end{align*}
\end{theorem}
\begin{proof}
First, note that, if $P_{Y_i^{(t)}|X_i^{(t)}}$ is BSEC with erasure probability $p_t$ and crossover probability $q_t$,
then, conditioned on $V_i^{(t)}=v \in \{0,2\}$, the channel $P_{Y_i^{(t)}|X_i^{(t)}V_i^{(t)}}(\cdot |\cdot,v)$ is
BSC with crossover probability $q_t$. Furthermore, we can verify that the emulated channel $P_{Y_i^{(t+1)}|X_i^{(t+1)}}$
in Step \ref{protocol-recursive-step6} is BSEC with erasure probability $p_{t+1}$ and crossover probability $q_{t+1}$.

In each round $t$, by our choice of $m_t$ and $n_t$, the probability of abortion in Step \ref{protocol-recursive-step3}
can be arbitrarily small for sufficiently large $n$. We can also verify that sending the index sets $({\cal I}_0^{(t)}, {\cal I}_1^{(t)},{\cal I}_2^{(t)})$
does not leak any information about $B$ exactly in the same manner as Protocol \ref{protocol-BSEC}.
Furthermore, by our choice of $\kappa_t$ and $l_t$, ${\cal P}_2$ can recover the key
$S_B^{(t)}$ with small error probability by the result on information reconciliation, and 
the key $S_{\overline{B}}^{(t)}$ is not leaked to ${\cal P}_2$ by the result on privacy amplification.
Note that the keys created in other rounds are independent, and the accumulated error probability
and leakage are the summation of those in each round.

Finally, the lower bound on the OT capacity follows from our choice of $l_t$ and 
that the rate of protocol is $\frac{1}{n}\sum_{t=1}^T l_t$.
\end{proof}

Note that the lower bound in Theorem \ref{theorem:recursive} subsumes \eqref{eq:OT-capacity-BSEC} when $T=1$;
\eqref{eq:OT-capacity-BSC-1round} when $p_1=0$ and $T=2$;\footnote{Technically speaking, since $m_1=0$ when $p_1=0$, 
we should skip the first round ($t=1$) to avoid the abortion in Step \ref{protocol-recursive-step3}.} 
and \eqref{eq:OT-capacity-BSEC-2round} when $T=2$.
The lower bounds on the OT capacity in Theorem \ref{theorem:recursive} for $T=1, 2,3$ are compared in Fig.~\ref{Fig:comparison}.
For comparison, we also plotted the upper bound in \eqref{eq:upper-1}.\footnote{Since it is computationally infeasible to 
evaluate the minimization for auxiliary alphabet of size $|{\cal J}|=|{\cal X}||{\cal Y}|+2 = 8$,
we only evaluated for $|{\cal J}|=2$, which is still a valid upper bound on the OT capacity.}
We find that, for crossover probabilities around $0.5$, the improved lower bound almost matches with the upper bound;
however, even for a larger value of $T$, we numerically verified that there is a gap between the lower bound and the upper bound.

\begin{figure}[tb]
\centering{
\includegraphics[width=0.9\linewidth]{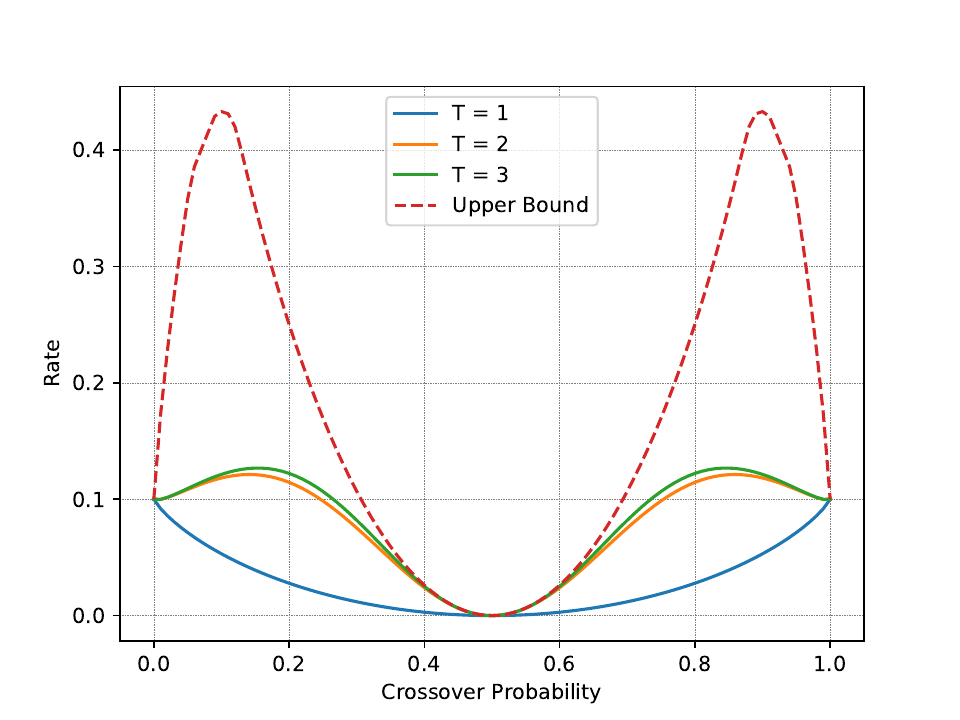}
\caption{A comparison of the lower bound on the OT capacity in Theorem \ref{theorem:recursive} 
for $T=1,2,3$ and the upper bound on the OT capacity in \eqref{eq:upper-1}, where the horizontal axis is
$0 \le q_1 \le 1$ and $p_1=0.1$ is fixed.}
\label{Fig:comparison}
}
\end{figure}

\section{Discussions} \label{section:discussion}

In this paper, we have proposed an OT protocol that recursively emulates erasure events using interactive
communication. Compared to the standard protocol in Protocol \ref{protocol-BSEC},
our protocol, Protocol \ref{protocol-recursive}, uses interactive communication in $T$ rounds.
Even though we exclusively considered the BSC and the BSEC in this paper, the idea
of recursively emulating erasure can be applied to more general channels.

\begin{protocol}[h] 
\caption{Protocol for Example \ref{example:necessity-interaction}} \label{protocol-example}
\begin{algorithmic}[1]
\STATE \label{protocol-example-step1}
${\cal P}_1$ transmit the uniform input $(X_{i,1},X_{i,2})$ for $1 \le i \le n$.

\STATE \label{protocol-example-step2}
For fixed margin $\Delta>0$, let ${\cal I}^{(1)}_{\san{e},1}$ be the first $n\big(\frac{1}{2}-\Delta\big)$ indices 
of $\{ i \in [n] : Y_{i,1}\neq \san{e}, Y_{i,2} = \san{e} \}$ and ${\cal I}^{(1)}_{\san{e},2}$ be the first $n\big(\frac{1}{2}-\Delta\big)$ indices 
of $\{ i \in [n] : Y_{i,1} = \san{e}, Y_{i,2} \neq \san{e} \}$.
Then, ${\cal P}_2$ sends $ (\hat{{\cal I}}^{(1)}_{\san{e},1}, \hat{{\cal I}}^{(1)}_{\san{e},2}) = ({\cal I}^{(1)}_{\san{e},1}, {\cal I}^{(1)}_{\san{e},2})$ if $B=0$
and  ${\cal P}_2$ sends $(\hat{{\cal I}}^{(1)}_{\san{e},1}, \hat{{\cal I}}^{(1)}_{\san{e},2})= ({\cal I}^{(1)}_{\san{e},2}, {\cal I}^{(1)}_{\san{e},1})$ if $B=1$ 
to ${\cal P}_1$.
Let ${\cal I}^{(1)}_2$ be the first $n\big(\frac{1}{2}-\Delta\big)$ indices of $\{i\in [n]: Y_{i,1}\neq \san{e}, Y_{i,2}\neq \san{e}\}$
(if there do not exist enough indices, abort the protocol).
Then, ${\cal P}_2$ also sends ${\cal I}^{(1)}_2$ to ${\cal P}_1$.

\STATE \label{protocol-example-step3}
Let $S_0$ be the concatenation of $(X_{i,1} : i \in \hat{{\cal I}}^{(1)}_{\san{e},1})$ and 
$(X_{i,2} : i \in \hat{{\cal I}}^{(1)}_{\san{e},2})$, and let $S_1$ be the concatenation of
$(X_{i,2} : i \in \hat{{\cal I}}^{(1)}_{\san{e},1})$ and $(X_{i,1} : i \in \hat{{\cal I}}^{(1)}_{\san{e},2})$.
Then, ${\cal P}_1$ sends $\Pi^{(1)}_0 = K^{(1)}_0 \oplus S^{(1)}_0$ and $\Pi^{(1)} = K^{(1)}_1 \oplus S^{(1)}_1$ to ${\cal P}_2$.

\STATE \label{protocol-example-step4}
Let $\hat{S}^{(1)}_B$ be the concatenation of $(Y_{i,1} : i \in {\cal I}^{(1)}_{\san{e},1})$ 
and $(Y_{i,2} : i \in {\cal I}^{(1)}_{\san{e},2})$. Then, ${\cal P}_2$ recovers $\hat{K}^{(1)} = \Pi^{(1)}_B \oplus \hat{S}^{(1)}_B$.

\STATE \label{protocol-example-step5}
For each $i \in {\cal I}^{(1)}_2$, ${\cal P}_1$ reveals the parity $X_{i,1} \oplus X_{i,2}$ to ${\cal P}_2$.
If the parity is $1$, ${\cal P}_1$ and ${\cal P}_2$ flip the value of $X_{i,1}$ and $Y_{i,1}$, respectively.
By relabeling $(0,0) \mapsto 0$, $(1,1) \mapsto 1$,
$(0,1) \mapsto \mathsf{e}$, and $(1,0)\mapsto \mathsf{e}$, the parties create
$(X^{(2)}_i : i \in {\cal I}^{(1)}_2)$ and $(Y^{(2)}_i : i \in {\cal I}^{(1)}_2)$, respectively. 

\STATE \label{protocol-example-step6}
${\cal P}_2$ sets ${\cal I}^{(2)}_B$ as the first $n\big(\frac{1}{4}-\Delta \big)$ indices of $\{ i \in {\cal I}^{(1)}_2 : Y^{(2)}_i \neq \san{e} \}$
and ${\cal I}^{(2)}_{\overline{B}}$ as the first $n\big(\frac{1}{4}-\Delta \big)$ indices of $\{ i \in {\cal I}^{(1)}_2 : Y^{(2)}_i = \san{e} \}$
(if there do not exist enough indices, abort the protocol).
Then, ${\cal P}_2$ sends $({\cal I}^{(2)}_0, {\cal I}^{(2)}_1)$ to ${\cal P}_2$.

\STATE \label{protocol-example-step7}
For $b=0,1$, 
${\cal P}_1$ sets $S^{(2)}_b = (X^{(2)}_i : i \in {\cal I}^{(2)}_b)$ and sends $\Pi^{(2)}_b = K^{(2)}_b \oplus S^{(2)}_b$ to ${\cal P}_1$.

\STATE \label{protocol-example-step8}
By setting $\hat{S}^{(2)}_B = (Y^{(2)}_i : i \in {\cal I}^{(2)}_B)$, ${\cal P}_2$ recovers 
$\hat{K}^{(2)} = \Pi^{(2)}_B \oplus \hat{S}^{(2)}_B$.
\end{algorithmic}
\end{protocol}

Even though it is not clear if our improved lower bound on the OT capacity is tight or not,
it seems necessary to use interactive communication in multiple rounds to attain the OT capacity in general.
As a further evidence for potential necessity of interactive communication in multiple rounds, 
let us consider the following simple example.
\begin{example} \label{example:necessity-interaction}
For ${\cal X} = \{0,1\}^2$ and ${\cal Y} = \{0,1,\san{e}\}^2$, let $W$ be the channel given by
\begin{align*}
& W(x_1,\san{e}|x_1,x_2) = W(\san{e},x_2|x_1,x_2) = W(x_1,x_2|x_1,x_2) = \frac{1}{4}, \\
&W(x_1,x_2 \oplus 1|x_1,x_2) = W(x_1 \oplus 1, x_2 |x_1,x_2) = \frac{1}{8}.
\end{align*}
For this channel, we can consider Protocol \ref{protocol-example} that is based on a similar idea as Protocol \ref{protocol-recursive}. 
Roughly, the parties use the erasure events $(*,\san{e})$ or $(\san{e},*)$ to realize OT at rate $\frac{1}{2}$ in the first phase;
when neither $(*,\san{e})$ nor $(\san{e},*)$ occur, by revealing the parity, the parties emulate erasure channel
with erasure probability $\frac{1}{2}$ to realize OT at rate $\frac{1}{2}\cdot \frac{1}{2}=\frac{1}{4}$ in the second phase.
In total, the OT capacity is lower bounded by $\frac{3}{4}$. In fact, for this channel, it can be
verified that the maximum of $I(X \wedge Y)$ is $\frac{3}{4}$, which is attained by the uniform input distribution.
Thus, by combining with the upper bound \eqref{eq:upper-2}, the OT capacity of this channel is 
$C_{\mathtt{OT}}(W) = \frac{3}{4}$.

It is unlikely that the OT capacity of this channel can be attained with
only 1 round of communication. If ${\cal P}_1$ reveals the parity $X_{i,1}\oplus X_{i,2}$
for an index such that $(Y_{i,1},Y_{i,2}) \in \{(*,\san{e}), (\san{e},*)\}$, that index become useless for OT;
thus, ${\cal P}_1$ should reveal the parity only after that ${\cal P}_2$ announces $(Y_{i,1},Y_{i,2}) \notin \{(*,\san{e}), (\san{e},*)\}$.
On the other hand, when $(Y_{i,1},Y_{i,2}) \notin \{(*,\san{e}), (\san{e},*)\}$, 
${\cal P}_2$ cannot recognize the emulated erasure event unless ${\cal P}_1$ reveals the parity.
Thus, it seems that $2$ rounds of interaction is unavoidable. 
In order to rigorously prove the necessity of $2$ rounds interaction, 
we need an upper bound on the OT capacity that is tailored for $1$ round protocols;
currently, all the known upper bounds do not take into account the number of rounds of OT protocols. 
\end{example}

\section*{Acknowledgment}

This work was supported in part by the Japan Society for the Promotion of Science (JSPS)
KAKENHI under Grant 20H02144, 23H00468, and 23K17455.

\newpage

\bibliographystyle{../../../09-04-17-bibtex/IEEEtran}
\bibliography{../../../09-04-17-bibtex/reference.bib}
\end{document}